\documentclass[11pt]{article}
\usepackage{amsmath,amsthm,amssymb,listings,xcolor}
\date{}
\newtheorem{theorem}{Theorem}
\newtheorem{lemma}{Lemma}

\newtheorem{proposition}{Proposition}

\newtheorem{open problem}{Open problem}

\begin{document}

\title{The Covering Radius of the Reed--Muller Code RM(1,9)  is 242}
\author{Qichun~Wang \footnote{School of Computer and Electronic Information/School of Artificial Intelligence,
Nanjing Normal University, Nanjing, China. Email: qcwang@fudan.edu.cn.}, \ Yueying~Lou\footnote{School of Computer and Electronic Information/School of Artificial Intelligence,
Nanjing Normal University, Nanjing, China. Email: 232212001@njnu.edu.cn.
} \ and \ Zifan~Wang\footnote{Nanjing Foreign Language School, Nanjing, China. Email: 114064796@qq.com.}}
\maketitle

\begin{abstract}
 What is the covering radius of the Reed--Muller code  RM(1,9) ? This is a long-standing open problem in
cryptography and coding theory.  In this paper, we resolve this problem, and prove that the covering radius of RM(1,9) is 242.  We
exclude nonlinearity 244 by obtaining a cubic Boolean function achieving a sharp Walsh $\ell^1$ bound. Its equality cases reduce either to an impossible cubic hyperplane restriction or to two homogeneous cubic forms in the Brier--Langevin classification. To exclude nonlinearity 243, we deduce the inequality ${\Vert W_g\Vert }_{1}+2|W_g(0)|\le 10368$, and then reduce either to a spectrum bounded by 24 or to the same expression to be at least 10496.
\end{abstract}



\section{Introduction}
Boolean functions play an important role in both cryptography and coding theory \cite{Car,Logachev}. The covering radius of RM(1,$m$) is the maximum nonlinearity of $m$-variable Boolean functions, which is known for $m$ even. In 1980, Mykkeltveit proved that the covering radius of RM(1, 7) is 56 \cite{Mykkeltvei}. Hou then gave a simpler proof for this in \cite{Hou}. However, despite more than four decades of effort, the covering radius of RM(1,9)  has yet to be resolved.

The general upper bound for RM(1, 9) is $244$. In 1996, Hou proved that the covering radius of RM(1, 9) in RM(4, 9) is 240 \cite{Hou2}. For $r\ge 5$, the covering radius of RM(1, 9) in RM($r$, 9) is also still unknown. There are 9-variable  Boolean functions with nonlinearity 242 in the generalized rotation symmetric class \cite{Kavut,Kavut2}, thus the covering radius of RM(1,9) is between 242 and 244.

For odd $m > 9$, the covering
radius of RM(1,$m$) is completely open. Asymptotically, the lower bound for the covering
radius of RM(1,$m$)  is essentially as large as the classical upper bound \cite{Sch}. For more results on the covering radius of RM($r, m$), we refer to \cite{Dough,Elim,Gao0,Gao,Gillot,Wang}.

The purpose of this paper is to determine the covering radius of RM(1,9), and prove the following theorem.
\begin{theorem}\label{thm:main}
For every 9-variable Boolean function $f$,
\begin{equation}\label{eq:main}
nl(f)\le242.
\end{equation}
Equivalently, the covering radius of  RM(1,9) is 242.
\end{theorem}

We exclude nonlinearity 244 through obtaining a cubic Boolean function achieving a sharp Walsh $\ell^1$ bound with two equality cases. One case is eliminated by eight-variable restrictions. In the other, a residue duality imposes two explicit conditions on the homogeneous cubic part. The classification of Brier and Langevin \cite{BrierLangevin,ClassificationData} leaves only two nonzero forms.  For both forms, the
reciprocity equations reduce them to finitely many normal forms,
and each normal form has a Walsh coefficient of absolute value at least 80.

For nonlinearity 243, we deduce the inequality ${\Vert W_g\Vert }_{1}+2|W_g(0)|\le 10368$, and then either lead back to the previously excluded spectrum, or give a numerical contradiction with a gap of 128.

The rest of the paper is organized as follows. In Section 2, the necessary
background is established. We obtain a cubic Boolean function achieving a sharp Walsh $\ell^1$ bound in Section 3, and exclude nonlinearity 244 in Section 4. In Section 5, we exclude nonlinearity 243 and prove the main theorem. We end in Section 6 with conclusions.

\section{Preliminaries}
Let $\mathbb{F}_{2}^{m}$ be the $m$-dimensional vector space over the binary finite field $\mathbb{F}_{2}$. A Boolean function of $m$ variables is a map from $\mathbb{F}_{2}^{m}$ to $\mathbb{F}_{2}$. We denote by $\mathbb{B}_{m}$ the set
of all $m$-variable Boolean functions.

We identify a Boolean function $f$ with its unique squarefree polynomial over $\mathbb{F}_2$; its degree, denoted by $\deg(f)$, is the degree of that polynomial. A Boolean function is affine if its degree is at most one. A ``cubic'' function throughout this paper has degree at most three.

For $I\subseteq[m]=\{1,\ldots,m\}$, write $x^I=\prod_{i\in I}x_i$. Polynomial sums are over $\mathbb{F}_2$ unless an integer-valued expansion is explicitly indicated.

 Let $f\in \mathbb{B}_m$. the Hamming weight of $f$, denoted by $wt(f)$, is the number of inputs at which $f$ takes the value 1. The function $f$ is  balanced if $wt(f)=2^{m-1}$.

Let $f, g\in \mathbb{B}_m$. The Hamming distance between $f$ and $g$ is $d(f,g)=wt(f+g)$. The nonlinearity of $f$, denoted by $nl(f)$, is its minimum Hamming distance from the set of affine Boolean functions.

For a real function $f$ on $\mathbb{F}_2^m$, the Fourier-Hadamard transform of $f$ is the integer-valued function over $\mathbb{F}_{2}^m$
defined by
\[
H_mf(\omega)=\sum_{x\in \mathbb{F}_{2}^m}f(x)(-1)^{\omega \cdot x}.
\]
Writing $(f_1,f_2)=\sum_{x}f_1(x)f_2(x)$, character orthogonality gives $H_m^2=2^mI$, $\sum_{\omega}(H_mf)(\omega)^2=2^m\sum_{\omega}f(\omega)^2$, and $(H_mf_1,f_2)=(f_1,H_mf_2)$ \cite{Car}.

The Walsh transform of $f\in \mathbb{B}_m$ is the function $W_f=H_m(-1)^f$. Parseval gives
\[
\sum_{a}W_f(a)^2=2^{2m}.
\]
Clearly, the nonlinearity of $f$ can be determined by
\[
nl(f)=2^{m-1}-\frac{1}{2}{\Vert W_f\Vert }_{\infty}.
\]
By the Parseval formulae, every 9-variable Boolean function has a ${\Vert W_f\Vert }_{\infty}$ at least 24, and its nonlinearity is upper bounded by 244.

 The Reed--Muller code RM$(d,m)$ consists of the truth tables of all $m$-variable Boolean functions of degree at most $d$. The covering radius of a binary code $C$ is the maximum, over all binary words of its length, of their minimum Hamming distance to $C$.
It is known that RM$(d,m)$ has minimum Hamming distance $2^{m-d}$ \cite{Car,Kasami1,MC}, which we state as the following lemma.
\begin{lemma}\label{lem:minimum-weight}
An $m$-variable nonzero Boolean function of degree at most $d$ has weight at least $2^{m-d}$, where $d\le m$.
\end{lemma}

Below, inner products and norms are unnormalized sums over their indexing sets.

\section{Reducing to a cubic form problem}
\begin{lemma}\label{lem:residue}
Let $f\in \mathbb{B}_m$. Its Walsh transform has the integer expansion
\begin{equation}\label{eq:integer-expansion}
W_f(a)=\sum_{I\subseteq[m]}(-2)^{|I|}a^I
                 \sum_{x:\,x_i=1\ (i\in I)}(-1)^{f(x)}.
\end{equation}
In particular, for $m=9$, if $W_f\equiv c\pmod{2^r}$ and $2\le r\le4$, then
\begin{equation}\label{eq:residue-degree}
\deg\left(\frac{W_f-c}{2^r}\bmod2\right)\le r-1.
\end{equation}
\end{lemma}
\begin{proof}
Clearly,
\[
W_f(a)=\sum_x(-1)^{f(x)+a\cdot x}=\sum_x(-1)^{f(x)}\prod_{i=1}^m(1-2a_ix_i).
\]
Expanding $\prod_i(1-2a_ix_i)$ proves \eqref{eq:integer-expansion}. For nonempty $I\subset[9]$, the inner sum of \eqref{eq:integer-expansion} is even, so the corresponding coefficient is divisible by $2^{|I|+1}$.
If $W_f-c$ is pointwise divisible by $2^r$, then all its integer multilinear coefficients are divisible by $2^r$. By \eqref{eq:integer-expansion}, for $m=9$ and $r\le 4$, after division, every coefficient of degree at least $r$ is even, and the result follows.
\end{proof}

\begin{lemma}\label{lem:divisibility}
If $f\in \mathbb{B}_m$ with $\deg(f)\le 3$, then
\begin{equation}\label{eq:cubic-div}
2^{\lceil m/3\rceil}\mid W_f(a)\quad(a\in \mathbb{F}_2^m).
\end{equation}
For a cubic $g\in \mathbb{B}_9$:
\begin{enumerate}
\item $8\mid W_g$, and $W_g\bmod16$ is constant;
\item $W_g\bmod32$ has integer multilinear degree at most 1, with nonconstant coefficients divisible by 16;
\item $W_g\bmod64$ has integer multilinear degree at most 3.
\end{enumerate}
For a cubic on 8 variables, its Walsh transform is divisible by 8 and is constant modulo 16.
\end{lemma}
\begin{proof}
Apart from a constant sign, $(-1)^{f(x)+a\cdot x}$ is a product of factors $1-2M(x)$, where each $M$ is a monomial of degree at most 3. A term selecting $j$ factors and involving $e$ distinct variables contributes to the sum over $x$ a multiple of $2^{j+m-e}$. Since $e\le\min(3j,m)$,
\[
j+m-e\ge\max(j,m-2j)\ge\lceil m/3\rceil.
\]
This proves \eqref{eq:cubic-div}, including restrictions of $f$.

Apply this fact to each restricted function in \eqref{eq:integer-expansion}. A coefficient indexed by $I$ is divisible by
\begin{equation}\label{eq:coefficient-div}
2^{\,|I|+\lceil(m-|I|)/3\rceil}.
\end{equation}
For $m=9$, the exponent is at least 4 when $|I|\ge1$, at least 5 when $|I|\ge2$, and at least 6 when $|I|\ge4$. For $m=8$, it is at least 3 at $I=\varnothing$ and at least 4 when $|I|\ge1$.
\end{proof}

\begin{proposition}\label{lem:forced}
Let $f\in \mathbb{B}_9$ be with ${\Vert W_f\Vert }_{\infty}=24$. Then its spectrum satisfies
\begin{equation}\label{eq:forced}
W_f\equiv8\pmod{16},\qquad W_f(a)\in\{\pm8,\pm24\},
\end{equation}
with 64 coefficients of magnitude 8 and 448 of magnitude 24. Moreover,
\begin{equation}\label{eq:def-g}
g=(W_f-8)/16\bmod2
\end{equation}
is cubic, and
\begin{equation}\label{eq:correlation}
((-1)^f, W_g)=(W_f, (-1)^g)=-10240.
\end{equation}
\end{proposition}
\begin{proof}
Clearly, $4\mid W_f$. By Lemma~\ref{lem:residue}, $W_f/4\bmod2$ is affine. If it is constantly one, every magnitude is at most 20. If it is balanced, the mean square is at most $(20^2+24^2)/2=488<512$. Hence it is zero, and $8\mid W_f$.

Now $q=W_f/8\bmod2$ is quadratic. If $q\ne1$, Lemma~\ref{lem:minimum-weight} applied to $1+q$ gives at least 128 zeros of $q$. On them $|W_f|\le16$, so
\[
\sum_a W_f(a)^2\le128\cdot 16^2+384\cdot24^2=253952<2^{18}.
\]
Thus $q=1$, proving \eqref{eq:forced}. Parseval fixes the two multiplicities. Lemma~\ref{lem:residue} makes \eqref{eq:def-g} cubic. Pointwise, $(-1)^gW_f=8$ on the 64 small positions and $(-1)^gW_f=-24$ elsewhere. Their sum is $-10240$, and the result follows.
\end{proof}

By Propostion~\ref{lem:forced}, if $f\in \mathbb{B}_9$ and ${\Vert W_f\Vert }_{\infty}=24$, then $g\in \mathbb{B}_9$ is cubic, and ${\Vert W_g\Vert }_{1}=10240$. The next result controls both the size of a cubic Walsh transform in $\ell^1$ and the form of its equality cases.

\begin{theorem}\label{thm:l1}
If $g\in \mathbb{B}_9$ is cubic, then
\begin{equation}\label{eq:l1}
{\Vert W_g\Vert }_{1}\le10240.
\end{equation}
Equality holds only in one of the following two cases:
\begin{enumerate}
\item[(A)] $W_g(x)\in\{\pm16,\pm48\}$ for every $x$;
\item[(B)] on an affine hyperplane $E$, $W_g(x)=\pm16$; off $E$, exactly 192 entries are $\pm32$ and the other 64 are zero.
\end{enumerate}
\end{theorem}

\begin{proof}
By Lemma~\ref{lem:divisibility}, $8\mid W_g$ and $W_g\bmod16$ is constant. Suppose first that $W_g\equiv8\pmod{16}$. By Lemma~\ref{lem:divisibility}, $W_g/8\equiv1+2\ell\pmod4$ for an affine Boolean function $\ell$. Then we have
\begin{equation}\label{eq:z}
z_x=(-1)^{\ell(x)}W_g(x)/8\equiv1\pmod4.
\end{equation}
Walsh inversion and Parseval imply
\begin{equation}\label{eq:z-moments}
\sum_xz_x=\pm64,\qquad \sum_xz_x^2=4096.
\end{equation}
For integers $z\equiv1\pmod4$, define
\begin{equation}\label{eq:majorant}
\Phi(z)=\frac{3z^2-2z+15}{16}.
\end{equation}
Then $|z|\le\Phi(z)$. Indeed, for positive $z$ the difference has numerator $3(z-1)(z-5)$; for negative $z$ its numerator is $(3z+5)(z+3)$. Both are nonnegative on the specified progression. Therefore,
\begin{equation}\label{eq:sum-majorant}
\sum_x\Phi(z_x)\le
\frac{3\cdot4096+128+15\cdot512}{16}=1256,
\end{equation}
and ${\Vert W_g\Vert }_{1}=8\sum_x |z_x|\le10048<10240$.

It remains that $16\mid W_g$. Put $k=W_g/16$. Then by the Parseval formulae and Lemma \ref{lem:divisibility},
\begin{equation}\label{eq:k-properties}
\sum_xk_x^2=1024,\qquad k\bmod2\text{ is affine}.
\end{equation}
For even integers $k$ and odd integers $j$, we have
\begin{equation}\label{eq:parity-majorants}
|k|\le k^2/2,\qquad |j|\le(j^2+3)/4.
\end{equation}
If the parity is constantly even, these inequalities give ${\Vert W_g\Vert }_{1}\le16\cdot1024/2=8192$. If it is constantly odd, we have $\sum|k_x|\le(1024+3\cdot512)/4=640$, with equality only at $k_x\in\{\pm1,\pm3\}$, which is case (A).

If the parity is nonconstant, then exactly 256 entries are odd. Let $A$ be the squared sum on those entries. Then $A\ge256$, and by (\ref{eq:parity-majorants}),
\begin{equation}\label{eq:balanced-l1}
\sum_x|k_x|\le\frac{A+3\cdot256}{4}+\frac{1024-A}{2}
=704-\frac A4\le640.
\end{equation}
Equality requires all odd entries to be $\pm1$ and all even entries to be in $\{0,\pm2\}$. Equation~\eqref{eq:k-properties} gives 192 nonzero even entries and 64 zeros, and the result follows. 
\end{proof}

\section{Excluding nonlinearity 244}
Throughout this section, we assume that $f$ has Walsh transform satisfying ${\Vert W_f\Vert }_{\infty}=24$.
Theorem~\ref{thm:l1} and \eqref{eq:correlation} force equality in \eqref{eq:l1}. In addition,
\begin{equation}\label{eq:negative-sign}
(-1)^{f(x)}=-\text{sgn} (W_g(x))\quad\text{whenever }W_g(x)\ne0.
\end{equation}

\subsection{Eliminating case (B) of Theorem~\ref{thm:l1}}
We eliminate case (B) of Theorem~\ref{thm:l1} first.

\begin{lemma}\label{lem:hyperplane}
The function $f$ has no restriction of degree at most three to an affine hyperplane of $V$.
\end{lemma}
\begin{proof}
An affine change of variables preserves the set of absolute Walsh values, so let the hyperplane be a coordinate slice. Let $A$ and $B$ be the 8-variable Walsh transforms of its two restrictions. For each frequency,
\begin{equation}\label{eq:restrictions}
A+B,\ A-B\in\{\pm8,\pm24\}, \ \ \ \  |A|+|B|\le 24.
\end{equation}
If the first restriction is cubic, by Lemma~\ref{lem:divisibility}, $A\equiv0$ or $8\pmod{16}$.

In the first case, \eqref{eq:restrictions} gives $A\in\{0,\pm16\}$. Parseval forces $|A|=16$ at every frequency, then by \eqref{eq:restrictions}, $|B|=8$ everywhere, impossible.

In the second case, $A\in\{\pm8,\pm24\}$. Parseval gives 96 entries of magnitude 24 and 160 of magnitude eight. Equation~\eqref{eq:restrictions} forces $B=0$ on the former and $|B|\le16$ on the latter. Thus $\sum B^2\le160\cdot16^2=40960<256^2$, again impossible.
\end{proof}

\begin{proposition}\label{prop:no-B}
Case (B) of Theorem~\ref{thm:l1} cannot arise from $f$.
\end{proposition}
\begin{proof}
Let $k=W_g/16$ and let $q=k\bmod2$ be its affine Boolean parity, so the affine hyperplane $E=\{q=1\}$. The Boolean function
\begin{equation}\label{eq:rr}
r=(k-q)/2\bmod2
\end{equation}
has degree at most three. Indeed, Lemma~\ref{lem:divisibility} ( the last assertion) gives an integer expansion of $k\bmod4$ of degree at most three. An affine Boolean function has integer degree at most two modulo 4, as follows from expanding its parity expression. Finite differences permit coefficientwise division of $k-q$ by two, proving the claim.

On $E$, $k=\pm1$, hence by \eqref{eq:rr}, $k=(-1)^r$. Equation~\eqref{eq:negative-sign} now gives $f=1+r$ on $E$, contradictory to Lemma~\ref{lem:hyperplane}.
\end{proof}

\subsection{Eliminating case (A) of Theorem~\ref{thm:l1}}

In case (A), define
\begin{equation}\label{eq:def-h}
h=(W_g/16-1)/2\bmod2.
\end{equation}
By Lemma~\ref{lem:divisibility}, $h$ is cubic. For $k=W_g/16\in\{\pm1,\pm3\}$, direct substitution in \eqref{eq:negative-sign} gives $k=(-1)^{h+1}-2(-1)^f$. Applying the Fourier-Hadamard transform $H_9$ yields
\begin{equation}\label{eq:second-transform}
W_h=-32(-1)^g-2W_f.
\end{equation}
The right side equals $-48(-1)^g$ at the 64 positions with $|W_f|=8$, and $16(-1)^g$ at the other positions. Therefore,
\begin{equation}\label{eq:residue-involution}
(W_h/16-1)/2\bmod2=g.
\end{equation}
Thus the residue operation in \eqref{eq:def-h} takes $g$ to $h$ and then returns to $g$.

For a homogeneous cubic $P=\sum_{|I|=3}p_I X^I$, define another homogeneous cubic $J(P)$ by (the coefficient of $X^I$ in $J(P)$)
\begin{equation}\label{eq:J}
[X^I]J(P)=\sum_{\substack{\{A,B\}:\ A\cup B=[9]\setminus I\\|A|=|B|=3}}p_Ap_B,
\end{equation}
and define
\begin{equation}\label{eq:C}
C(P)=\sum_{\substack{\{A,B,D\}:\ A\cup B\cup D=[9]\\|A|=|B|=|D|=3}}p_Ap_Bp_D.
\end{equation}
Both sums are over unordered partitions and take values in $\mathbb{F}_2$.

\begin{lemma}\label{lem:JC}
If $P$ is the homogeneous cubic part of the function $g$ obtained above, then
\begin{equation}\label{eq:JC}
J(J(P))=P,\qquad C(P)=0.
\end{equation}
\end{lemma}
\begin{proof}
For $|I|=3$, the coefficient of $a^I$ in the integer expansion of $W_g$ is $-8\sum_{x_i=1\ (i\in I)}(-1)^{g(x)}$. To obtain the cubic coefficient of $h$ in \eqref{eq:def-h}, divide by 32 and reduce modulo 2. This is, up to an irrelevant sign, the 6-variable restricted Walsh sum divided by 4, modulo 2.

In the product expansion used in Lemma~\ref{lem:divisibility}, the only contributions to that restricted sum modulo 8 select two disjoint cubic monomials covering the 6 free variables. Therefore the homogeneous cubic part of $h$ is exactly $J(P)$. Applying the same argument to \eqref{eq:residue-involution} proves $J(J(P))=P$.

For a 9-variable Walsh sum modulo 16, the only contributing selections consist of three disjoint cubic monomials covering all 9 variables. Thus, for every frequency $x$,
\begin{equation}\label{eq:C-residue}
W_g(x)/8\equiv C(P)\pmod2.
\end{equation}
Since $16\mid W_g$ in case (A), we have $C(P)=0$.
\end{proof}

The classification of $RM(3,9)/RM(2,9)$ under $GL(9,2)$ is completely known. So, using \eqref{eq:JC}, we can reduce the problem to a finite verification.
In fact, Brier and Langevin classify the 349 orbits of $RM(3,9)/RM(2,9)$
under $GL(9,2)$ \cite{BrierLangevin}. We use their published representative list \cite{ClassificationData}, with its original indexing, including the zero form at index 349.

 For $A\in GL(9,2)$, replace $g(a)$ by $g(Aa)$ and $f(x)$ by $f(A^{-T}x)$. Their transforms change according to
\begin{equation}\label{eq:linear-action}
W_f'(a)=W_f(Aa),\qquad W_g'(x)=W_g(A^{-T}x).
\end{equation}
All preceding spectral and residue identities remain valid. We first choose $A$ so that the cubic part is a listed representative, and then apply Lemma~\ref{lem:JC} to the transformed functions.

\begin{proposition}\label{prop:filter}
Among the 349 listed representatives, $J(J(P))=P$ holds precisely at indices
\[
282,\ 290,\ 320,\ 326,\ 329,\ 349.
\]
Imposing $C(P)=0$, only the representatives with indices $282,290,349$ survive.
\end{proposition}
\begin{proof}
This is the only one time finite calculation. Represent a monomial by the 9-bit mask of its support. For each representative, evaluate \eqref{eq:J} by switching the complementary mask for every unordered disjoint pair of present cubic monomials. Evaluate \eqref{eq:C} by counting modulo two the unordered partitions into three present monomials. Apply the first operation twice and compare all 84 cubic coefficients. Both supplied implementations perform these operations for all 349 entries and obtain the stated lists.
\end{proof}

The zero class 349 is also impossible. A quadratic function of rank $2s\le8$ has Walsh $\ell^1$ norm $512\cdot2^s\le8192$, whereas \eqref{eq:correlation} requires 10240. To verify the quadratic formula, reduce its quadratic part by a nonsingular linear substitution to $\sum_{i=1}^s x_{2i-1}x_{2i}$. Its Walsh transform is a product of $s$ two-variable transforms and $9-2s$ one-variable transforms. Adding affine terms only shifts or negates the spectrum.

With $ijk$ denoting $x_ix_jx_k$, the two remaining forms are
\begin{align}
P_{282}={}&125+145+345+346+137\notag\\
          &+248+349+569+789,\label{eq:P282}\\
P_{290}={}&125+235+245+345+136+127+347\notag\\
          &+128+238+148+129+349+569+789.\label{eq:P290}
\end{align}

We eliminate these two cases in Appendix A, and prove the following theorem.
\begin{theorem} \label{66}
Let $g$ be a 9-variable cubic Boolean function satisfying
$
W_g(a)\equiv16\pmod{32}
$
for every $a\in \mathbb{F}_2^9$, and suppose that
\[
T^2(g)=g,
\qquad
T(g)(a)=\frac{W_g(a)-16}{32}\pmod2.
\]
Then the cubic part of $g$ cannot belong to class $282$ or class $290$ if
$
\max_a|W_g(a)|\le48.
$
Indeed, in either class one necessarily has
$
{\max_a|W_g(a)|\ge80.}
$
\end{theorem}

We now can conclude that there does not exist a 9-variable Boolean function with nonlinearity 244.

\begin{theorem}\label{thm:no24}
No Boolean function on 9 variables has nonlinearity 244.
\end{theorem}
\begin{proof}
Proposition~\ref{lem:forced} and Theorem~\ref{thm:l1} force case (A) or (B), with \eqref{eq:negative-sign}. Proposition~\ref{prop:no-B} excludes (B). In (A), Lemma~\ref{lem:JC} and Proposition~\ref{prop:filter} leave the zero form or the two forms \eqref{eq:P282} and \eqref{eq:P290}. The zero form was excluded by its quadratic Walsh norm. Theorem~\ref{66} excludes both remaining forms.
\end{proof}

\section{Excluding Nonlinearity 243}\label{sec:26}

\begin{theorem}\label{thm:weighted}
For every cubic $g\in \mathbb{B}_9$, we have
\begin{equation}\label{eq:weighted}
{\Vert W_g\Vert }_{1}+2|W_g(0)|\le10368.
\end{equation}
\end{theorem}
\begin{proof}
In the case $W_g\equiv8\pmod{16}$, use \eqref{eq:z}--\eqref{eq:majorant}. On the progression $z\equiv1\pmod4$, we have
$
3|z|-\Phi(z)\le12.
$
Indeed, for positive $z$ the numerator is $-3z^2+50z-15$, whose maximum on the progression is 192 at $z=9$. For negative $z$, the numerator is $-3z^2-46z-15$, with maximum 160 at $z=-7$. Therefore, by \eqref{eq:z} and \eqref{eq:sum-majorant},
\[
\frac{{\Vert W_g\Vert }_{1}+2|W_g(0)|}{8}=\sum_{x\neq0} |z_x|+3|z_0|
\le\sum_x\Phi(z_x)+3|z_0|-\Phi(z_0)\le1268.
\]
This gives the stronger bound 10144 in this case.

If $16\mid W_g$, put $s=|k_0|$ and apply \eqref{eq:parity-majorants} to the other entries. When the parity is constantly odd,
\begin{align*}
3s+\sum_{x\ne0}|k_x|
&\le 3s+\frac{\sum_{x\ne0}(k_x^2+3)}{4}\le3s+\frac{2^{10}+3\cdot512-(s^2+3)}{4}\\
&=\frac{2593-(s-6)^2}{4}\le648,
\end{align*}
since $s$ is odd. When it is constantly even,
\[
3s+\sum_{x\ne0}|k_x|
\le\frac{1033-(s-3)^2}{2}\le516,
\]
since $s$ is even.

For balanced parity and odd $s$, let $A\ge255$ be the squared sum on the other odd entries. Then
\begin{align*}
3s+\sum_{x\ne0}|k_x|&=3s+\sum_{x\ne0,\ k_x\ odd}|k_x|+\sum_{k_x\ even}|k_x|\\
&\le3s+\frac{A+765}{4}+\frac{1024-s^2-A}{2}\\
&\le644-\frac{(s-3)^2}{2}\le644.
\end{align*}
For balanced parity and even $s$, all 256 odd entries remain. Their squared sum is at least 256, giving
\[
3s+\sum_{x\ne0}|k_x|
\le3s+\frac{1280-s^2}{2}
=\frac{1289-(s-3)^2}{2}\le644.
\]
Thus the largest upper bound is $16\cdot648=10368$.
\end{proof}

We now can prove the main theoreom by eleminating nonlinearity 243.

\begin{proof}[Proof of Theorem~\ref{thm:main}]
In view of Theorem~\ref{thm:no24}, it suffices to exclude $nl(f)=243$, or equivalently ${\Vert W_f\Vert }_{\infty}=26$. Now, all Walsh coefficients have a common residue 2 modulo 4. By Lemma~\ref{lem:residue},
\[
\ell=(W_f-2)/4\bmod2
\]
is affine. Multiplication of $W_f(a)$ by $(-1)^{\ell(a)}$ corresponds to an affine transformation of $f$, and normalizes the residue to $W_f\equiv2\pmod8$. Thus we may assume
\begin{equation}\label{eq:seven-values}
W_f(a)\in\{-22,-14,-6,2,10,18,26\}.
\end{equation}

Again by Lemma~\ref{lem:residue}, $q=(W_f-2)/8\bmod2$ is quadratic. Clearly, the seven values in \eqref{eq:seven-values} obey
\begin{equation}\label{eq:q-majorant}
W_f(a)^2\le4W_f(a)+252+320q(a).
\end{equation}

If $q\ne1$, Lemma~\ref{lem:minimum-weight} applied to $1+q$ gives $wt(q)\le384$. Summing \eqref{eq:q-majorant}, we have
\[
262144\le4\cdot512+252\cdot512+320\cdot384=253952,
\]
which is impossible. Hence $q=1$ and
\begin{equation}\label{eq:shifted-values}
W_f\in\{-22,-6,10,26\},\qquad
T:=W_f-2\in\{\pm8,\pm24\}.
\end{equation}

If $f(0)=0$, change the value to $1$ and leave all other entries fixed. The resulting Boolean function has Walsh transform $T$, contrary to Theorem~\ref{thm:no24}.

Therefore, $f(0)=1$, and
\begin{equation}\label{eq:T-energy}
\sum_aT(a)^2=2^{18}+2^{10}+4\cdot 2^9=2^{18}+8\cdot 2^9.
\end{equation}
Together with \eqref{eq:shifted-values}, this gives exactly 56 entries of magnitude 8 and 456 of magnitude 24.

Define
\begin{equation}\label{eq:g26}
g=(W_f-10)/16\bmod2=(T-8)/16\bmod2.
\end{equation}
Then $g$ is cubic by Lemma~\ref{lem:residue}. On the 56 small positions, $(-1)^gT=8$; on the other 456 positions, $(-1)^gT=-24$. Thus
\begin{equation}\label{eq:10496}
({(-1)^g},{T})=56\cdot 8-456\cdot 24=-10496.
\end{equation}
Since $W_f=T+2$ and $\sum_a(-1)^{g(a)}=W_g(0)$, symmetry of $H_9$ gives
\[
((-1)^f,W_g)=(W_f,(-1)^g)=-10496+2W_g(0).
\]
Theorem~\ref{thm:weighted} now yields the contradiction
\[
10496=2W_g(0)-((-1)^{f},{W_g})\le2|W_g(0)|+{\Vert W_g\Vert }_1\le10368.
\]
Hence nonlinearity 243 is impossible, and the result follows.
\end{proof}

\section{Conclusions}
This paper determined the covering radius of the Reed--Muller Code RM(1,9).
That is, the maximum nonlinearity of 9-variable Boolean functions is 242. Since the known 9-variable functions with nonlinearity 242 are not balanced, it is still an open problem to determine whether there exists a balanced 9-variable Boolean function with nonlinearity 242, which could be further considered in the future.

\section*{Acknowledgment}

The authors would like to thank the financial support from the National Natural Science Foundation of China (Grant 62172230).

\begin{appendix}
\section{Eliminating both remaining classes 282 and 290}
Define $T(g)(a)=\frac{W_g(a)-16}{32}\pmod2$. We have $T^2(g)=g$.
In the following, we will prove that classes $282$ and $290$ are impossible.

For a cubic Boolean function $p$ satisfying $W_p(a)\equiv16\pmod{32}$,  summing
its defining congruences for $T(p)$ gives  (the coefficient of $x_ix_j$ in $T(p)$)
\begin{equation}
[x_ix_j]T(p)=
\frac{wt(p|_{x_i=x_j=0})}{4}\pmod2,
\label{eq:quadratic-coeff}
\end{equation}
and similarly  (the coefficient of $x_ix_jx_k$ in $T(p)$)
\begin{equation}
[x_ix_jx_k]T(p)=
\frac{wt(p|_{x_i=x_j=x_k=0})}{2}\pmod2.
\label{eq:cubic-coeff}
\end{equation}
All coefficient calculations below may be obtained from the inclusion-exclusion
identity
\begin{equation}
wt\left(\sum_{M\in\mathcal M}x_M\right)
=
\sum_{\varnothing\ne\mathcal S\subseteq\mathcal M}
(-2)^{|\mathcal S|-1}
2^{m-\left|\bigcup_{M\in\mathcal S}M\right|}.
\label{eq:weight-formula}
\end{equation}

\subsection{Class 282: the determinant model}

Let $X=(x_{ij})\in M_3(F_2)$,
and $D(X)=\det X$.
This cubic represents class $282$. Indeed, the transformation is
\[
(y_1,\ldots,y_9)=(  t_1,\ t_2+t_6,\ t_9,\ t_2,\ t_8,
t_3+t_4+t_8,\ t_5,\ t_7,\ t_3),
\]
where $t_1,\ldots,t_9$ are the row-major entries of $X\in M_3(F_2)$.
For $A,B\in M_3(F_2)$ define
\begin{align}
R_A(X)&=\sum_{i,j}A_{ij}\prod_{k\ne j}x_{ik},\\
C_B(X)&=\sum_{i,j}B_{ij}\prod_{k\ne i}x_{kj}.
\end{align}
By a translation and addition of an affine function, every function in this class has the form
\begin{equation}
g(X)=G_{A,B}(X):=D(X)+R_A(X)+C_B(X).
\label{eq:det-normal-form}
\end{equation}
The nine nonconstant coefficients of $W_g/16\bmod2$ give nine independent
linear conditions on the 36 quadratic coefficients. In determinant coordinates,
the resulting kernel is the 18-dimensional space represented by $R_A+C_B$.

Write $\text{cof}(A)$ for the cofactor matrix, without transposing it. A direct
application of \eqref{eq:weight-formula} gives
\begin{equation}
\frac{W_{G_{A,B}}(Y)}{16}
\equiv
1+\det A+\det B+\sum_{i,j}A_{ij}B_{ij}
\pmod2,
\label{eq:det-walsh}
\end{equation}
independently of $Y$. Hence,
\begin{equation}
\det A+\det B+A:B=0,
\label{eq:det-constraint}
\end{equation}
where $A:B=\sum_{ij}A_{ij}B_{ij}$.

On the quadratic parameters, $T$ induces the map
\begin{equation}
\Phi(A,B)=
\bigl(B+\text{cof}(A),\ A+\text{cof}(B)\bigr).
\label{eq:det-Phi}
\end{equation}
Thus $T^2(g)=g$ implies $\Phi^2(A,B)=(A,B)$.
Define
\[
U\star V=\text{cof}(U+V)+\text{cof}(U)+\text{cof}(V).
\]
Using $\text{cof}(\text{cof}(A))=(\det A)A$,
we obtain
\begin{align}
(\det A)A+B\star \text{cof}(A)&=0,\label{eq:det-recip1}\\
(\det B)B+A\star \text{cof}(B)&=0.\label{eq:det-recip2}
\end{align}

Suppose $A$ is invertible. Under
$
(A,B)\longmapsto(PAQ,P^{-T}BQ^{-T})
$
we may normalize $A=I$. Then \eqref{eq:det-recip1} becomes
$
I+B\star I=0.
$
But
\begin{equation}
B\star I=(\text{tr}(B))I+B^T,
\end{equation}
whose trace is zero in characteristic two, whereas $\text{tr}(I)=1$. This is impossible.
Therefore $A$ is singular; by symmetry, so is $B$.
Hence
\begin{equation}
\operatorname{rank} A,\operatorname{rank} B\le2.
\label{eq:det-rank}
\end{equation}
Since both determinants vanish, \eqref{eq:det-constraint} gives
$
A:B=0.
$

Transposition exchanges $(A,B)$ with $(B^T,A^T)$, so we may assume
$\operatorname{rank} A\ge\operatorname{rank} B$.

If $\operatorname{rank} A=0$, then $A=B=0$.
If $\operatorname{rank} A=1$, normalize $A=E_{11}$. Since $\operatorname{rank} B\le1$ and
$A:B=0$, we have $B_{11}=0$. Up to the stabilizer of $E_{11}$ the nonzero
possibilities of $B$ are $E_{12}$, $E_{21}$ and $E_{22}$.

If $\operatorname{rank} A=2$, normalize $A=D_2:=\text{diag}(1,1,0)$.
Equations \eqref{eq:det-recip1}-\eqref{eq:det-recip2} force
\begin{equation}
B=
\begin{pmatrix}
0&0&u_1\\
0&0&u_2\\
v_1&v_2&b
\end{pmatrix},
\qquad
u_1v_1+u_2v_2=0.
\label{eq:B-rank2}
\end{equation}
The stabilizer of $D_2$ reduces this to
$
B\in
\{0,E_{33},E_{13},E_{31},E_{13}+E_{32}\}.
$
Thus, up to the allowed equivalences, the complete list is
\begin{equation}
\begin{gathered}
(0,0),\quad (E_{11},0),\quad(E_{11},E_{12}),\quad(E_{11},E_{21}),\quad(E_{11},E_{22}),\\
(D_2,0),\quad(D_2,E_{33}),\quad(D_2,E_{13}),\quad(D_2,E_{31}),\quad(D_2,E_{13}+E_{32}).
\end{gathered}
\label{eq:det-list}
\end{equation}

\begin{table*}[!h]
\centering
\caption{Values of $W_{G_{A,B}}(Y)$.} \label{tab:enumeration1}
% ``Odd residue'' counts the cases in which $W_{P+Q}/16$ is constantly odd. Retained spectra have $W_{P+Q}(x)\in\{\pm16,\pm48\}$ for every $x$; for each of these, $F_Q=-\text{sgn}  (W_{P+Q})$.}\label{tab:enumeration
\setlength{\tabcolsep}{5pt}
\begin{tabular}{@{}crrrrrr@{}}
\hline
$A$ & $B$ & $Y$ & $W_{G_{A,B}}(Y)$\\
\hline
$0$ & $0$ & $0$ & $176$\\
$E_{11}$ & $0$ & $0$ & $112$\\
$E_{11}$ & $E_{12}$ & $0$ & $80$\\
$E_{11}$ & $E_{21}$ & $0$ & $80$\\
$E_{11}$ & $E_{22}$ & $0$ & $112$\\
$D_2$ & $0$ & $0$ & $80$\\
$D_2$ & $E_{33}$ & $0$ & $80$\\
$D_2$ & $E_{13}$ & $0$ & $80$\\
$D_2$ & $E_{31}$ & $E_{12}$ & $80$\\
$D_2$ & $E_{13}+E_{32}$ & $E_{12}$ & $80$\\
\hline
\end{tabular}
\end{table*}

For these ten normal forms, direct Walsh evaluation gives Table 1. Every entry in the table is therefore an elementary finite character sum.
In particular,
\begin{equation}
\max_Y|W_{G_{A,B}}(Y)|\ge80
\end{equation}
for every reciprocal form in class $282$.
Thus class $282$ is impossible whenever the Walsh spectrum is restricted to
$\{\pm16,\pm48\}$.

\subsection{Class 290: the Hermitian determinant model}

Let $\mathbb{F}_4=\mathbb{F}_2(\omega)$, $\omega^2=\omega+1$, and $\bar z=z^2$.
Consider
\[
X=
\begin{pmatrix}
 a&z&\bar y\\
 \bar z&b&x\\
 y&\bar x&c
\end{pmatrix},
\qquad a,b,c\in \mathbb{F}_2,\quad x,y,z\in \mathbb{F}_4.
\]
Its determinant is
\begin{equation}
N(X)=abc+a\,x\bar x+b\,y\bar y+c\,z\bar z+\text{Tr}_{4/2}(xyz).
\label{eq:herm-det}
\end{equation}
Its cubic part represents class $290$. Indeed, the transformation is
\begin{align*}
(y_1,\ldots,y_9)=(&t_3+t_4,\ t_5+t_7,\ t_5,\ t_4+t_5+t_6+t_7,t_2+t_3+t_5,\\
& t_1+t_2+t_5+t_7,\ t_4+t_6+t_8+t_9,\ t_5+t_7+t_9,\ t_3),
\end{align*}
where
$(a,b,c)=(t_1,t_2,t_3)$, $x=t_4+\omega t_5$,\ $y=t_6+\omega t_7$ and $z=t_8+\omega t_9$.

Use the nondegenerate symmetric pairing
$
(X ,Y)=\text{tr}(XY)\in \mathbb{F}_2.
$
For $A\in M_3(\mathbb{F}_4)$ define
\begin{equation}
Q_A(X)=
\text{Tr}_{4/2}\left(\sum_{i,j}A_{ij}\prod_{k\ne j}X_{ik}\right).
\label{eq:herm-Q}
\end{equation}
After the allowed normalization, every function in this class satisfying the Walsh congruence has the form
$
g=N+Q_A.
$

Put
$t=\text{tr}(A)$, $s=\text{tr}(\text{cof} (A))$, and $d=\text{det}(A)$.
A direct expansion gives
\begin{equation}
\frac{W_{N+Q_A}(Y)}{16}
\equiv
1+t\bar t+\text{Tr}_{4/2}(d)
\pmod2,
\label{eq:herm-walsh}
\end{equation}
so
\begin{equation}
t\bar t=\text{Tr}_{4/2}(d).
\label{eq:herm-constraint}
\end{equation}

The induced transformation on $A$ is
$
\Phi(A)=\overline{\text{cof}(A)}+tI.
$
The condition $T^2(g)=g$ implies $\Phi^2(A)=A$. Using
$
\text{adj}(A)=A^2+tA+sI,
$
we obtain
\begin{equation}
 t\,A^2+(d+1+t\bar t)A+\bar s\,I=0.
\label{eq:herm-recip}
\end{equation}

If $t=0$, then \eqref{eq:herm-constraint} gives $d\in\{0,1\}$.
If $d=0$, equation \eqref{eq:herm-recip} gives $A=\bar sI$, and taking traces
forces $s=0$, hence $A=0$. If $d=1$, equation \eqref{eq:herm-recip} gives
$s=0$, so the characteristic polynomial
\[
\chi_A(\lambda)=\lambda^3+1=(\lambda+1)(\lambda+\omega)(\lambda+\omega^2),
\]
and therefore
$
A\sim \text{diag}(1,\omega,\omega^2).
$

If $t\ne0$, then $t\bar t=1$, so
$
d\in\{\omega,\omega^2\}.
$
Equation \eqref{eq:herm-recip} becomes
$
\bar t A^2+dA+\bar sI=0.
$
Its annihilating polynomial has nonzero derivative $d$, so $A$ is semisimple.
It cannot be scalar because a nonzero scalar $3\times3$ matrix has determinant $1$.
Thus its eigenvalues have multiplicities $2,1$, and after the allowed scalar and
similarity normalizations we obtain
\[
A\sim \text{diag}(\omega,\omega,1),
\qquad
A\sim \text{diag}(\omega^2,\omega^2,1).
\]

Thus only the four forms
\begin{equation}
0,
\quad
\text{diag}(1,\omega,\omega^2),
\quad
\text{diag}(\omega,\omega,1),
\quad
\text{diag}(\omega^2,\omega^2,1)
\label{eq:herm-four}
\end{equation}
remain.

Set
\[
Y_1=
\begin{pmatrix}
1&\omega&\omega^2\\
\omega^2&1&1\\
\omega&1&1
\end{pmatrix},
\qquad
Y_2=
\begin{pmatrix}
1&0&1\\
0&1&1\\
1&1&1
\end{pmatrix}.
\]
For $A=\text{diag}(r,s,t)$, put
\[
\psi(x,y,z)=
\text{Tr}_{4/2}\bigl(xyz+r z\bar y+s\bar z x+t y\bar x
+x\bar u+y\bar v+z\bar w\bigr),
\]
where $(u,v,w)=(Y_{23},Y_{31},Y_{12})$.
Summing first over $a,b,c$ gives
\begin{equation}
W_{N+Q_A}(Y)=8S_1-2S_0,
\label{eq:herm-char-sum}
\end{equation}
where
\begin{align}
S_1&=
\sum_{\substack{x\bar x=\epsilon_1\\y\bar y=\epsilon_2\\z\bar z=\epsilon_3}}
(-1)^\psi,\\
S_0&=
\sum_{x,y,z\in \mathbb{F}_4}
(-1)^{\psi+x\bar x+y\bar y+z\bar z+\epsilon_1+\epsilon_2+\epsilon_3},
\end{align}
and each $\epsilon_i\in \mathbb{F}_2$ is fixed. For the four normal forms in \eqref{eq:herm-four}, with frequencies
$E_{11},Y_1,Y_2,Y_2$, respectively, direct calculation gives
\[
(S_1,S_0)=(3,-28),\quad(7,-12),\quad(11,4),\quad(11,4).
\]
Hence in every case $W_{N+Q_A}(Y)=8S_1-2S_0=80$.
The elementary identity
\[
\sum_{z\in \mathbb{F}_4}
(-1)^{\text{Tr}_{4/2}(Rz)+z\bar z}
=2-4\cdot1_{R=0}
\]
provides a direct reduction of the character sums.
Therefore
\begin{equation}
\max_Y|W_{N+Q_A}(Y)|\ge80
\end{equation}
for every reciprocal form in class $290$.
Thus class $290$ is impossible whenever the Walsh spectrum is contained in
$\{\pm16,\pm48\}$.

So, we have proved Theorem \ref{66}.
\end{appendix}

\begin{thebibliography}{00}



\bibitem{BrierLangevin}
E. Brier and P. Langevin, ``Classification of Boolean Cubic Forms of Nine Variables,'' in \emph{Proc. 2003 IEEE Information Theory Workshop}, 2003, pp. 179--182.

\bibitem{Car} C. Carlet, ``Boolean Functions for Cryptography and Coding Theory,'' Cambridge University Press, 2021.

\bibitem{ClassificationData}
E. Brier and P. Langevin, ``Cubics in Nine Variables''. [Online]. Available: https://langevin.univ-tln.fr/project/cubics/.

\bibitem{Dough} R. Dougherty, R. D. Mauldin, and M. Tiefenbruck, ``The Covering Radius of the Reed--Muller Code \(RM(m-4,m)\) in \(RM(m-3,m)\),'' \emph{IEEE Transactions on Information Theory}, 68:1 (2022), 560--571.

\bibitem{Elim} D. Elimelech, H. Wei, and M. Schwartz, ``On the Generalized Covering Radii of Reed--Muller Codes,'' \emph{IEEE Transactions on Information Theory}, 68:7 (2022), 4378--4391.

\bibitem{Gao0} J. Gao, ``Numerical Results and Asymptotic Lower Bound on the Covering Radius of Reed--Muller Codes \(RM(2,11)\) and \(RM(3,n)\),'' \emph{IEICE Transactions on Fundamentals}, E109-A:1 (2026), 35--46.

\bibitem{Gao} J. Gao, H. Kan, Y. Li and Q. Wang,  ``The Covering Radius of the Third-Order
Reed--Muller Code RM(3,7) is 20,'' \emph{IEEE Transactions on Information Theory}, 69:6 (2023), 3663--3671.

\bibitem{Gillot} V. Gillot and P. Langevin, ``Covering Radius of \(RM(4,8)\),'' \emph{Advances in Mathematics of Communications}, 18:2 (2024), 383--393.

\bibitem{Hou} X. Hou,  ``Covering Radius of the Reed--Muller Code R(1, 7) - A Simpler Proof,'' \emph{Journal of Combinatorial Theory, Series A}, 74:2 (1996), 337--341.

\bibitem{Hou2} X. Hou,  ``The Covering Radius of R(1, 9) in R(4, 9), " \emph{Designs, Codes and Cryptography}, 8:3 (1996), 285--292.

\bibitem{Kasami1} T. Kasami and N. Tokura, ``On the Weight Structure of the Reed--Muller
Codes,'' \emph{IEEE Transactions on Information Theory}, 16 (1970), 752--759.

\bibitem{Kavut} S. Kavut and M. D. Y{\"u}cel,  ``9-variable Boolean Functions with Nonlinearity 242 in the Generalized Rotation Symmetric Class,'' \emph{Information and Computation}, 208:4 (2010), 341--350.

\bibitem{Kavut2} S. Kavut, S. Maitra and M. D. Y{\"u}cel,  ``Search for Boolean Functions with Excellent Profiles in the Rotation Symmetric Class,'' \emph{IEEE Transactions on Information Theory}, 53:5 (2007), 1743--1751.

\bibitem{Logachev} O. A. Logachev, A. A. Salnikov and V. V. Yashchenko, ``Boolean Functions in Coding Theory and Cryptography,'' American Mathematical Society, 2012.

\bibitem{Mykkeltvei} J. Mykkeltveit,  ``The Covering Radius of the (128,8) Reed--Muller Code is 56 (Corresp.),'' \emph{IEEE Transactions on Information Theory}, 26:3 (1980), 359--362.

\bibitem{MC} F. J. MacWilliams and N. J. Sloane, ``The Theory of Error-correcting Codes,''
North Holland, 1977.

\bibitem{Sch} K. Schmidt,  ``Asymptotically optimal Boolean functions,''\emph{Journal of Combinatorial Theory, Series A}, 164 (2019), 50--59.

\bibitem{Wang} Q. Wang, ``The Covering Radius of the Reed--Muller Code RM(2,7) is 40,'' \emph{Discrete Mathematics}, 342 (2019), Art. no. 111625.
\end{thebibliography}
\end{document}